\documentclass[11pt]{article}
\usepackage[margin=1.1in]{geometry}
\usepackage{microtype}
\usepackage{amsmath,amssymb,amsthm}
\usepackage{enumitem}
\usepackage{xcolor}
\usepackage{todonotes}
\usepackage{booktabs}
\usepackage{tikz}
\usepackage{verbatim}
\usepackage{cite}

\theoremstyle{plain}
\newtheorem{theorem}{Theorem}
\newtheorem{lemma}{Lemma}

\theoremstyle{definition}
\newtheorem{definition}{Definition}

\theoremstyle{remark}

\usetikzlibrary{arrows.meta,calc}

\usepackage[dvipsnames]{xcolor}

\usepackage{hyperref}
\hypersetup{colorlinks=true, urlcolor=Blue, citecolor=Green, linkcolor=BrickRed, breaklinks, unicode}

\usepackage[capitalize,nameinlink]{cleveref}
\hypersetup{
  pdftitle={A Tight Bound for Facial Distance Patterns in Planar Graphs},
  pdfauthor={Viktor Fredslund-Hansen, Shay Mozes, Oren Weimann},
  pdfsubject={Planar graph distance patterns, forbidden alternation, Okamura--Seymour metric compression, planar diameter}
}
\setlist[itemize]{leftmargin=2em,itemsep=0.25em,topsep=0.4em}
\setlist[enumerate]{leftmargin=2.2em,itemsep=0.25em,topsep=0.4em}

\newcommand{\bits}{\{-1,+1\}}

\newcommand{\cF}{\mathcal{F}}

\newcommand{\Otilde}{\tilde{O}}

\newcommand{\x}{\ensuremath p_\#}

\date{}

\title{A Tight Bound for Facial Distance Patterns in Planar Graphs}
\author{
Viktor Fredslund-Hansen
\and
Shay Mozes\thanks{Reichman University, Israel. Supported in part by Israel Science Foundation grant 365/26.}
\and
Oren Weimann\thanks{Haifa University, Israel. Supported in part by Israel Science Foundation grant 365/26.}
}
\begin{document}
\maketitle
\begin{abstract}
Let $G$ be an undirected unweighted planar graph and let $S=(s_0,\dots,s_{k-1})$ be the vertices of a designated face, listed in cyclic order. 
Consider a vector that stores the distances from an arbitrary vertex $v$ to all vertices of $S$. The {\em pattern} of $v$ is obtained by taking the difference between every pair of consecutive values in this vector.
Li and Parter [STOC'19] proved an upper bound of $O(k^3)$ on the number of unique patterns over all vertices of $G$.  
We improve this to $O(k^2)$, matching a known lower bound and settling a conjecture in [ISAAC'22]. 
The simple proof was found by OpenAI's GPT 5.6-Sol model. 

Plugging this new bound into known results has the following three immediate implications for undirected unweighted planar graphs: (1) it gives an improved  compression of the Okamura--Seymour metric, (2) it improves the space required by constant-time exact distance oracles, and (3) it improves the fastest distributed algorithm for computing the diameter.     
    We further present a previously unknown and nontrivial implication: 
    a (centralized) $\Otilde(n^{8/5})$-time algorithm for computing the diameter, improving over the $\Otilde(n^{5/3})$ algorithm of [SODA'18] which works for weighted directed planar graphs. 
Thus, there is currently a gap between the time for computing the diameter between weighted and unweighted planar graphs. 
\end{abstract}

\section{Introduction}\label{sec:intro}
Let $G$ be an undirected unweighted planar graph and let $S=(s_0,\dots,s_{k-1})$ be the vertices of a designated face, listed in cyclic order. The \emph{pattern} $p^v$ of a vertex $v$ in $G$ is the vector  $p^v = \left \langle d_G(v,s_1) - d_G(v,s_0) , d_G(v,s_2) - d_G(v,s_1) , \ldots , d_G(v,s_{k-1}) - d_G(v,s_{k-2}) \right \rangle$ where $d_G(x,y)$ denotes the $x$-to-$y$ distance in $G$. Let $\x$ denote the number of unique patterns over all vertices of $G$. 
Li and Parter \cite{LP19} proved that $\x=O(k^3)$, and used this bound to obtain an $\Otilde(D^5)$-round algorithm for computing the diamater in the CONGEST model of distributed computation. Their bound was later used by Fredslund-Hansen, Mozes, and Wulff-Nilsen~\cite{FHMWN21} to construct exact distance oracles with constant query time and subquadratic space for unweighted undirected planar graphs. Mozes, Wallheimer, and Weimann \cite{MWW22} gave a simple $\Omega(k^2)$ lower bound and conjectured that the upper bound should be $O(k^2)$. 
This problem and conjecture was recently mentioned by Da Wei Zheng as an important open problem in a Dastuhl seminar on Metric Sketching and Dynamic Algorithms for Geometric and Topological Graphs\cite[Open problem 5.6]{DagRep}.
We provide a simple proof of this conjecture, which we found by a single prompt to OpenAI’s GPT 5.6-Sol model.
It is surprising (not to say embarrasing) that this open problem has such a simple proof, which has eluded the community despite the human efforts invested in it.

\medskip
\noindent
{\bf The Li-Parter proof.}
Since the graph is unweighted and undirected, every entry of $p^v$ is in $\{-1,0,1\}$ by the triangle inequality.
The patterns can be transformed\footnote{This transformation was presented in \cite{MWW22} and in the conference talk of \cite{LP19}.} to be binary (i.e. over $\{-1,1\}$ instead of $\{-1,0,1\}$) vectors of dimension $2k-1$ by  subdividing every edge of the graph. 

\begin{definition}[Forbidden configuration]
    Two binary vectors $A,B$ (both over $\{-1,1\}$) are said to have a forbidden configuration if there are four indices $a<b<c<d$ such that $A_a = -1, A_b = 1, A_c = -1, A_d=1$ but $B_a = 1, B_b = -1, B_c = 1, B_d=-1$. 
\end{definition}

Li and Parter's proof is based on the simple observation that no two vectors ${p}^u,{p}^v$ can have a forbidden configuration. 
The reason is that such $(-1,1,-1,1),(1,-1,1,-1)$ induced patterns correspond to an illegal configuration of shortest paths in planar graphs.

Now, consider the patterns of all vertices of $G$ as the rows of a binary matrix $P$. The {\em VC-dimension} of $P$ is the largest number $d$ such that there exists in $P$ a submatrix of $d$ columns that contains all possible $2^d$ different rows. 
The above forbidden configuration implies that the VC-dimension of $P$ is at most $3$. By Sauer's Lemma~\cite{Sauer72}, this means that there are $O((2k-1)^3) = O(k^3)$ distinct rows. This is their entire proof.

\medskip
\noindent
{\bf Limitations of the Li-Parter proof.}
The following set of sequences over $\{-1,1\}^{k-1}$  was observed by \cite{MWW22}:
$
\left \{ (-1)^{x_1} \circ 1^{x_2} \circ (-1)^{x_3} \circ 1^{k-1-x_1-x_2-x_3} \; | \; x_1 + x_2 + x_3 < k\right \}.
$
There is no pair of sequences in this set that contains the forbidden $(1,-1,1,-1),(-1,1,-1,1)$ configuration, and yet its cardinality is $\Theta(k^3)$. This shows that VC-dimension alone is not enough to break the $O(k^3)$ upper bound.

\medskip
\noindent
{\bf The new proof.}
We prompted ChatGPT with a description of the state of the art and asked for any improvement on the upper bound. 
It replied that the answer is $\x = O(k^2)$. Specifically, it noted that since the sequence of vertices along the distinguished face is cyclic, the entries of the pattern (as a vector in $\bits^{2k-1}$) must sum up to either 1 or -1 because if we added one more term that closes the cycle back to $s_0$, the sum would telescope to 0. It follows that the number of $1$'s in any pattern is either $k$ or $k-1$. 
This property was previously observed, but has not been exploited.
It then pointed out that sets of vectors in $\bits^{2k-1}$ that avoid the forbidden configuration and have the same number of 1's are in fact what is known as {\em weakly separated subsets of $[2k-1]$}\footnote{Here, and throughout we use $[n]$ to denote the set $\{1, 2, \dots, n\}$. Each vector corresponds to the set of coordinates of its 1 entries.}. 
The notion of weakly separated subsets was defined and studied by Leclerc and Zelevinsky \cite{LZ98} in the context of quasi-commuting quantum flag minors. 
They proved (Theorem 1.2 in \cite{LZ98}) that the maximal size of a family of weakly separated subsets of $[n]$ is $\frac{1}{2}(n+2)(n+1)+1$. 
Their inductive proof is elementary and takes about one page.
Using their result in our case (i.e. applying their bound with $[n]=[2k-1]$ twice; once for the patterns with $k$ 1's and once for the patterns with $k-1$ 1's) immediately gives an upper bound of  $4k^2+k+2 = O(k^2)$ on $\x$.
Leclerc and Zelevinsky also proved a tight bound of $k(n-k)+1$ for weakly separated $k$-subsets of $[n]$\footnote{Substitute $\ell=k$ into the statement of Theorem 1.3 in \cite{LZ98}.}. The proof of this bound in \cite{LZ98} is longer and more complicated. Using this bound in our case yields a slightly better upper bound of $2k^2-2k+2$.
In addition, ChatGPT gave an elegant and simple proof that directly gives the stronger bound on the size of a set of vectors in $\bits^n$ that exclude the forbidden configuration, without going through the stronger definition of weakly separated subsets of $[n]$\footnote{Without cardinality constraints, two weakly separated subsets correspond to pattens that avoid the forbidden configuration, but not the other way around.}. We clarified, further simplified, and cleaned that proof and present it in \cref{section:proof}.   
 
 \medskip
\noindent
{\bf Implications.}
 There are several implications of this new $\x = O(k^2)$ bound: 
 \begin{itemize}
 	\item {\bf Okamura--Seymour metric compression \cite{OS81}.} This problem  asks to compactly encode the $S$-to-$T$ distances between the vartices of a face $S$ and an arbitrary subset of vertices $T$ of an undirected unweighted planar graph $G$. A query $(v,s_i)$ to the encoding (with $v\in T$ and $s_i \in S$) returns the $v$-to-$s_i$ distance in $G$.
 Li and Parter~\cite{LP19} (see also comment in \cite{MWW22})  presented a compression of size $\tilde O(k\cdot \x+|T|)$ and $O(1)$ query time. This was improved by Mozes, Wallheimer and Weimann \cite{MWW22} to $\tilde{O}(\x+|T|)$ space with $\tilde{O}(1)$ query time. By plugging $\x=O(k^3)$ their bound was $\tilde{O}(k^3+|T|)$. Plugging in the new $O(k^2)$ bound immediately improves this to $\tilde{O}(k^2+|T|)$.

 	\item {\bf Distance oracles with constant query time \cite{FHMWN21}.} This problem asks for a compact data structure that reports exact distances in $G$ (between any pair of vertices) in {\em constant} time. Fredslund-Hansen, Mozes and Wulff-Nilsen \cite{FHMWN21} showed how to use patterns to obtain such an oracle for undirected unweighted planar graphs with $O(n^{5/3+\epsilon})$ space, for any constant $\epsilon>0$. By plugging in $\x=O(k^2)$ their space immediately improves to $O(n^{3/2+\epsilon})$.  	 	
     \item {\bf Distributed diameter computation \cite{LP19}.}
This problem is in distributed computing, in the CONGEST model. It asks for the smallest number of rounds required to compute the diameter $D$ of an undirected unweighted planar graph. Li and Parter \cite{LP19} used patterns to obtain an algorithm requiring $\tilde{O}(D^5)$ rounds. 
Mozes, Wallheimer and Weimann \cite{MWW22} observed that the same algorithm can be easily tweaked to run in $\tilde{O}(D^4)$ rounds. By plugging in $\x=O(k^2)$ this immediately improves to $\tilde{O}(D^3)$ rounds. 
 
 	\item {\bf Centralized diameter computation}.
The current fastest algorithm to compute the diameter of a planar graph (by Gawrychowski, Kaplan, Mozes, Sharir and Weimann \cite{GKMSW18}) runs in $\Otilde(n^{5/3})$ time and  works even for weighted directed planar graphs. Improving its running time is a major open problem in planar graph algorithms. In \cref{section:diameter} we present a non-trivial use of the new $\x=O(k^2)$ bound to achieve a faster $\Otilde(n^{8/5})$ time algorithm for diameter in unweighted undirected planar graphs.
Thus currently there is gap between weighted and unweighted diameter computation times.
 \end{itemize}

\section{The ChatGPT Proof}\label{section:proof}
In this section we prove the main theorem:

\begin{theorem}\label{thm:main}
Let $G$ be a connected, undirected, unweighted planar graph with a designated face $f$. Let $k$ be the number of edges of the facial walk of $f$. Then $\x \in O(k^2)$.
\end{theorem}

 We first state and prove \cref{lem:uniform}, which bounds the maximum cardinality of a family of vectors in $\bits^n$ with exactly $r$ $1$'s.
 We then provide, for the sake of completeness, the already known arguments that imply the main theorem.

\begin{lemma}\label{lem:uniform}
Let $0\le r\le n$. The maximal size of a family $\cF$ of vectors in $\bits^n$, each with exactly $r$ 1's, that avoids the forbidden configuration is $r(n-r)+1$.
\end{lemma}

\begin{proof}
Assume that $1\le r\le n-1$, as otherwise the claim is trivial.

\medskip 
\noindent\textbf{Step 1: encode the family by a prefix trie.}
Let $T$ be the rooted prefix trie whose nodes are all prefixes of the vectors in $\cF$ (the root corresponds to the empty prefix), with children given by appending $1$ or $-1$ whenever the extended prefix occurs. All vectors have length $n$, so the leaves of $T$ are exactly the vectors of $\cF$, and the number of leaves is $L=|\cF|$.
Call an internal node \emph{branching} if it has both children, and let $b$ and $u$ be the numbers of branching and single-child internal nodes. The tree has $u+b+L$ nodes, hence $u+b+L-1$ edges, and counting edges by out-degree gives $u+2b$. Therefore
\begin{equation}\label{eqation:leaves-branches}
 L=b+1 ,
\end{equation}
and it suffices to prove $b\le r(n-r)$.

\medskip
\noindent\textbf{Step 2: map branching nodes into a rectangle.}
For a binary prefix $P$ set $x(P)=\#\{\text{1's in }P\}$ and $y(P)=\#\{\text{-1's in }P\}$. If $P$ is branching, its $-1$-child extends to a full vector with exactly $n-r$ -1's, so $y(P)+1\le n-r$; its $1$-child extends to a vector with exactly $r$ 1's, so $x(P)+1\le r$. Hence every branching node satisfies
\[
 0\le y(P)\le n-r-1,\qquad 0\le x(P)\le r-1 ,
\]
and the location map $\lambda(P)=(x(P),y(P))$ sends branching nodes into a rectangle of exactly $r(n-r)$ lattice points (see \cref{figure:trie}).\\

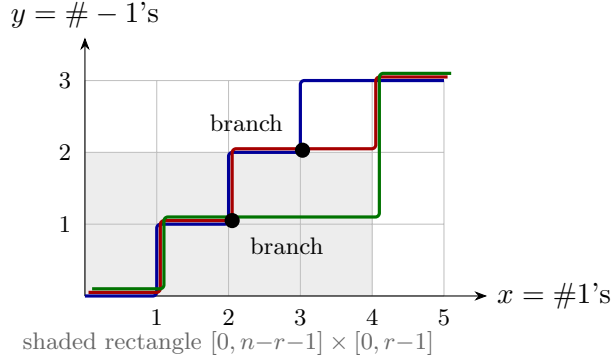
\begin{figure}[h!]
\centering
\begin{tikzpicture}[scale=0.95,>=Stealth]
  \fill[black!7] (0,0) rectangle (4,2);
  \draw[step=1, black!25, very thin] (0,0) grid (5,3);
  \draw[->] (0,0) -- (5.6,0) node[right] {$x=\#\text{1's}$};
  \draw[->] (0,0) -- (0,3.6) node[above] {$y=\#-1\text{'s}$};
  \foreach \x in {1,...,5} \node[below] at (\x,-0.05) {\footnotesize \x};
  \foreach \y in {1,...,3} \node[left] at (-0.05,\y) {\footnotesize \y};
  \draw[blue!60!black, very thick, rounded corners=1.5pt]
    (0.00,0.00) -- (1.00,0.00) -- (1.00,1.00) -- (2.00,1.00) -- (2.00,2.00) -- (3.00,2.00) -- (3.00,3.00) -- (5.00,3.00);
  \draw[red!65!black, very thick, rounded corners=1.5pt]
    (0.05,0.05) -- (1.05,0.05) -- (1.05,1.05) -- (2.05,1.05) -- (2.05,2.05) -- (3.05,2.05) -- (4.05,2.05) -- (4.05,3.05) -- (5.05,3.05);
  \draw[green!45!black, very thick, rounded corners=1.5pt]
    (0.10,0.10) -- (1.10,0.10) -- (1.10,1.10) -- (2.10,1.10) -- (3.10,1.10) -- (4.10,1.10) -- (4.10,2.10) -- (4.10,3.10) -- (5.10,3.10);
  \fill[black] (2.05,1.05) circle (3.0pt);
  \fill[black] (3.03,2.03) circle (3.0pt);
  \node[anchor=north west] at (2.16,0.96) {\footnotesize branch};
  \node[anchor=south east] at (2.90,2.16) {\footnotesize branch};
  \node[black!60] at (2.0,-0.66) {\footnotesize shaded rectangle $[0,n{-}r{-}1]\times[0,r{-}1]$};
\end{tikzpicture}
\caption{A family $\cF$ of vectors with $n=8$ and $r=5$ as monotone lattice paths ($1=$ right, $-1=$ up), sharing prefixes as in the trie; slight offsets distinguish coinciding segments. Branching prefixes (dots) are confined to the shaded $ (n-r)\times r$ rectangle, and the proof of \cref{lem:uniform} shows that no two branching prefixes occupy the same lattice point; hence $b\le r(n-r)$ and $|\cF|=b+1$.}
\label{figure:trie}
\end{figure}

\noindent\textbf{Step 3: \boldmath$\lambda$ is injective on branching nodes.}
Suppose two distinct branching nodes $P\ne Q$ satisfy $\lambda(P)=\lambda(Q)$; then both $P$ and $Q$ have the same number of 1's and $-1$'s, so they have the same length $\ell$. 
Let $j\le\ell$ be the last position at which they differ; interchanging $P$ and $Q$ if necessary, assume $P_j=-1$, $Q_j=1$. Since $P$ and $Q$ contain the same number of 1's and agree after position $j$, the excess $1$ at position $j$ is cancelled earlier: there is $i<j$ with $P_i=1$, $Q_i=-1$.
As $P$ is branching, its $1$-child occurs in the trie, so we can choose $A\in\cF$ that starts with $P$ followed by $1$. As $Q$ is branching, we can choose $B\in\cF$ that starts with $Q$ followed by -$1$. At coordinates $i,j,\ell+1$:
\[
\begin{array}{c|ccc}
 & i & j & \ell+1\\ \hline
 A & 1 & -1 & 1\\
 B & -1 & 1 & -1
\end{array}
\]
Consider the sum $\sum_{h\le n}\bigl(A_h-B_h\bigr)$. Since both $A$ and $B$ contain $r$ $1$'s and $n-r$ $-1$'s, the sum  is $0$. Since $P$ and $Q$ have equally many $1$'s and $-1$'s, the first $\ell$ coordinates $\sum_{h\le\ell}\bigl(A_h-B_h\bigr)$ contribute $0$ to the sum, while coordinate $\ell+1$ contributes $2$;  Hence
\begin{equation}\label{equation:suffix-balance}
 \sum_{h=\ell+2}^{n}\bigl(A_h-B_h\bigr)=-2 .
\end{equation}
If no coordinate $h>\ell+1$ had $A_h=-1$, $B_h=1$, then every summand in \eqref{equation:suffix-balance} would be $0$ or $2$, which cannot sum to $-2$. So such an $h$ exists, and the four indices $i,j,\ell+1,h$ form a forbidden configuration, a contradiction. 
Thus $\lambda$ is injective, $b\le r(n-r)$, and \eqref{eqation:leaves-branches} completes the proof.
\end{proof}

\begin{proof}[Proof of \cref{thm:main}]
We assume $G$ has no self loops since those do not affect distances. We also assume the designated face $f$ is simple. This can be done without loss of generality by the following standard transformation (see e.g. \cite{Klein06}):  
Let $W$ be the boundary cyclic walk of $f$. Let $s_0, s_2, \dots , s_{k-1}$ be the vertices of $f$ along $W$. Note that if $W$ is not a simple cycle then the $s_i$'s are not distinct. 
For each $i$, add a new distinct pendant vertex $v_i$ embedded inside $f$ with and edge $s_iv_i$ embedded between $s_{i-1}s_i$ and $s_is_{i+1}$ in the cyclic order around $s_i$ (indices are modulo $k$). 
Now connect all the new pendant vertices with edges $v_iv_{i+1}$. 
See \cref{fig:inflate}. This creates a new face $f'$ whose facial walk is a simple cycle of length $k$ through the newly added vertices. 
This construction does not change the distance between any two vertices of the original $G$. 
Furthermore, For any vertex $u$ of the original $G$ and any $0\leq i \leq \ell$, the $u$-to-$s_i$ distance in $G$ equals the $u$-to-$v_i$ distance in the modified graph. 
Hence, the number of patterns for $f$ is bounded by the number of patterns for the new simple face $f'$ in the modified graph.

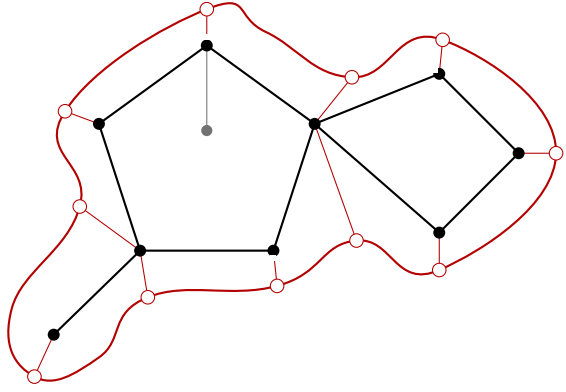
\begin{figure}[t]
\centering
\begin{tikzpicture}[scale=1.5,>=Stealth]
  \coordinate (a) at (90:1);
  \coordinate (b) at (162:1);
  \coordinate (c) at (234:1);
  \coordinate (d) at (306:1);
  \coordinate (e) at (18:1);
  \coordinate (g) at (2.05,0.75);
  \coordinate (h) at (2.75,0.05);
  \coordinate (i) at (2.05,-0.65);
  \coordinate (f) at (-1.35,-1.55);
  \coordinate (p) at (0,0.25);
  \coordinate (ca)  at (0,1.32);
  \coordinate (cb)  at (-1.25,0.42);
  \coordinate (cc1) at (-1.12,-0.42);
  \coordinate (cf)  at (-1.52,-1.92);
  \coordinate (cc2) at (-0.52,-1.22);
  \coordinate (cd)  at (0.62,-1.12);
  \coordinate (ce1) at (1.32,-0.72);
  \coordinate (ci)  at (2.05,-0.98);
  \coordinate (ch)  at (3.08,0.05);
  \coordinate (cg)  at (2.08,1.05);
  \coordinate (ce2) at (1.28,0.72);
  \draw[thick] (a)--(b)--(c)--(d)--(e)--(a);
  \draw[thick] (e)--(g)--(h)--(i)--(e);
  \draw[thick] (c)--(f);
  \draw[black!55] (a)--(p);
  \draw[red!70!black,thick] plot [smooth cycle,tension=0.8] coordinates
    {(ca) (cb) (cc1) (-1.72,-1.32) (cf) (-0.95,-1.75) (cc2)
     (cd) (ce1) (ci) (ch) (cg) (ce2) (0.52,1.12)};
  \foreach \s/\t in {ca/a, cb/b, cc1/c, cf/f, cc2/c, cd/d, ce1/e, ci/i, ch/h, cg/g, ce2/e}{
    \draw[red!70!black] (\s)--(\t);
  }
  \node[red!70!black] at (-0.68,1.42) {};
  \foreach \pt in {a,b,c,d,e,f,g,h,i}{ \fill (\pt) circle (1.5pt); }
  \fill[black!55] (p) circle (1.4pt);
  \foreach \pt in {ca,cb,cc1,cf,cc2,cd,ce1,ci,ch,cg,ce2}{ \draw[red!70!black,fill=white] (\pt) circle (1.7pt); }
  \node[above=2pt,  inner sep=1pt] at (a) {};
  \node[left=1.5pt] at (b) {};
  \node[below right=0.5pt] at (c){} ;
  \node[below=1.5pt, inner sep=1pt] at (d) {};
  \node[left=4pt] at (e){};
  \node[below right=0.5pt] at (f){};
  \node[above left=0.5pt,  inner sep=1pt] at (g) {};
  \node[below=2.5pt] at (h) {};
  \node[below left=0.5pt] at (i) {};
  \node[right=1.5pt,black!55] at (p) {};
\end{tikzpicture}
\caption{Turning the non simple infinite face $f$ (in black) into a simple face $f'$ (in red).}
\label{fig:inflate}
\end{figure}

We therefore continue under the assumption that $f$ is simple. 
We subdivide each edge of $G$ by introducing an artificial vertex in the middle of each edge. 
This makes the graph bipartite. In particular, the size of the distinguished face $f$ becomes $2k$, and the distance between each pair of original vertices of $G$ doubles. 
Since each pattern in the subdivided $G$ corresponds to exactly one pattern in the original $G$, it therefore suffices to bound the number of patterns in the subdivided $G$. 
Hence, for the remainder of the proof when we refer to $G$ we mean the subdivided $G$.
Since $G$ is unweighted undirected and bipartite, distances from any vertex $u$ to adjacent vertices on the distiguished face $f$ differ by exactly $1$. Hence, each pattern is a vector in $\bits^{2k-1}$.

Denote $s_k=s_0$. Then by cyclicity, for any vertex $v$, the sum $\sum_{i=1}^k d_G(v,s_i) - d_G(v,s_{i-1})$ is 0. 
Hence, by definition of the pattern $p^v$, summing all the entries of $p^v$ yields $\pm 1$, so the number of $1$ entries in $p^v$ is either $k-1$ or $k$.

Li and Parter \cite[Section 3]{LP19} proved that the set ${\mathcal P}$ of patterns avoids the forbidden configuration.
Hence, ${\mathcal P}$ can be partitioned into two disjoint sets of patterns $ {\mathcal P}_k \cup  {\mathcal P}_{k-1}$, where $ {\mathcal P}_k$ (resp., $ {\mathcal P}_{k-1})$ contains patterns $p$ with exactly $k$ 1's (resp., $k-1$ 1's) and satisfies \cref{lem:uniform} with $n=2k-1$ and $r=k$ (resp., $r=k-1$).  
Hence, invoking \cref{lem:uniform}, we get $\x = | {\mathcal P}| = | {\mathcal P}_k| + | {\mathcal P}_{k-1}| \leq (k(k-1)+ 1) + ((k-1)k + 1) = 2k^2 -2k+2 =  O(k^2)$. 
\end{proof}

\section{The Implication to Diameter Computation}\label{section:diameter}

The diameter of a directed weighted planar graph can be computed in truly subquadratic time: Cabello \cite{Cab19} gave the first such algorithm, and the current fastest is the $\Otilde(n^{5/3})$ time algorithm of Gawrychowski, Kaplan, Mozes, Sharir and Weimann \cite{GKMSW18}.  
For \emph{unweighted undirected} planar graphs no uniform improvement has been known. Abboud, Mozes and Weimann \cite{AMW23} gave algorithms with running times $\Otilde(nD^2)$ and $n^{3+o(1)}/D^2$, which beat $\Otilde(n^{5/3})$ only when the diameter $D$ is below $n^{1/3}$ or above $n^{2/3}$. 
In this section we show that the new quadratic pattern bound implies such a uniform improvement. 
Throughout this section $G$ is simple; parallel edges and loops do not affect distances and may be discarded.

\begin{theorem}[Unweighted planar diameter]\label{thm:diameter}
There is an $\Otilde(n^{8/5})$ time  algorithm for computing 
the diameter, the radius, and all vertex eccentricities of an undirected unweighted $n$-vertex planar graph.
\end{theorem}

\begin{proof}
Our algorithm incorporates the notion of patterns into the framework of \cite{Cab19,GKMSW18}. In that framework, first an {\em $r$-division} is performed in linear time \cite{ShayrDivision}. This is a partition of the graph $G$ into $O(n/r)$ {\em regions} where each region $R$ contains $O(r)$ vertices and $O(\sqrt{r})$ boundary vertices $\partial R$ (vertices shared by other regions). Then, for each region $R$, the {\em Voronoi diagram} data structure of \cite{GKMSW18} is built in $\Otilde(r^2)$ time (so $\Otilde(n\cdot r)$ time over all regions). A query to this data structure is a weight assignment $w(\cdot)$ to the boundary vertices $\partial R$. 
It returns, in $\Otilde(\sqrt{r})$ time, a list indicating for each boundary vertex $u \in \partial R$, the furthest vertex from $u$ in $R$ among all the vertices in $u$'s {\em Voronoi cell}.  
The Voronoi cell of $u$ contains all vertices $v \in R$ such that $u$ is the boundary vertex that minimizes the quantity $w(u)+d_R(u,v)$.

The Voronoi diagram data structure is used as follows: For every vertex $v$ of $G$, for every region $R$ of the $r$-division, a query is performed on the Voronoi diagram of $R$ with the weights $w(\cdot)$ being the distances in $G$ from $v$ to $\partial R$. This way, the  query reveals the furthest vertex from  $v$ in $G$ among all vertices of $R$ (and the maximum such value is returned as $G$'s diameter).\footnote{A special treatment is required to handle the specific region $R$ that contains $v$. It is possible that the path from $v$ to its furthest in $R$ does not touch $\partial R$. To handle this case, we simply run Dijkstra's algorithm from $v$ inside $R$ after initializing the vertices of $\partial R$ with the aforementioned $w(\cdot)$. Dijkstra takes $\Otilde(r)$ time
for every $v$ and so $\Otilde(n\cdot r)$ time
overall.} 
Computing the weights $w(\cdot)$ from $v$ to all boundary nodes of all regions takes $\Otilde(n/\sqrt{r})$ time \cite{FR06}. The queries performed for each $v$ require $\Otilde(\sqrt{r})$ time per region so $\Otilde(n/ \sqrt{r})$ time over all regions. The total running time is thus    
$\Otilde(n^2/ \sqrt{r}+ n\cdot r)$ which is  $\Otilde(n^{5/3})$ by taking $r=n^{2/3}$.

Our idea is to apply the Voronoi queries on $R$ not for every vertex of $G$ but only for every unique pattern in the graph $G\setminus R \cup \partial R$ with respect to the face $\partial R$.\footnote{We assume here that $\partial R$ lies on a single face of $G\setminus R \cup \partial R$. In reality, it may lie on a constant number of faces \cite{ShayrDivision}. Handling this case is standard and is described in \cite{GKMSW18}.} 
Specifically, among all vertices with the same pattern, we query the Voronoi diagram of $R$ just for the {\em representative} vertex $v$ that maximizes $d_{G\setminus R \cup \partial R}(v,s_0)$ where $s_0$ is an arbitrary chosen vertex of $\partial R$ that defines the starting vertex of the pattern.
The observation is (cf. \cite[Corollary 9]{FHMWN21}) that if two vertices have the same pattern with respect to distances in $G\setminus R \cup \partial R$, then they have the same pattern w.r.t. distances in all of $G$ as well. 
Hence, this choice of representative $v$ guarantees that $d_G(v,u)\le d_G(v',u)$ for every $u \in R$ and every $v'$ that has the same pattern as $v$. 
Thus, for the sake of computing the eccentricity of $u$, considering just the representative $v$ suffices.

 There is another issue however. While the vertices of $\partial R$ lie on a face of the the graph $G\setminus R \cup \partial R$, this face does not necessarily consist of just the $O(\sqrt r)$ vertices of $\partial R$, but may contain many (up to $O(r)$) additional vertices. In the weighted setting, this is easily handled by adding artificial edges of weight infinity that connect just the boundary vertices $\partial R$.  
 In the unweighted setting however, we cannot add these artificial weighted edges (as then the graph is no longer unweighted and so patterns do not apply). Instead, we will only use the (trivial) fact that there is a {\em facial walk} in $R$ (a walk that visits all vertices of $\partial R$) that is of length $k=O(r)$. See Lemma~3 in \cite{FHMWN21} and the preceding discussion there. 
  In other words, we use the $\x=O(k^2)$ bound on the number of unique patterns (\cref{thm:main}) with $O(k^2)=O(r^2)$, and not $O((\sqrt r)^2)$.
 
Formally, we find the unique patterns in 
the graph $G'$ obtained from $G$ by removing all vertices of $R$ that are not part of the facial walk of $R$ that contains the boundary vertices $\partial R$. 
This guarantees that indeed the boundary vertices $\partial R$ lie on a face of size $O(r)$ of $G\setminus R \cup \partial R$. 
The randomized algorithm of \cite{MWW22} can find all $O(k^2)=O(r^2)$ patterns in this graph in $\Otilde(n)$ time (and can easily find for each pattern the appropriate representative $v$).

For each unique pattern with representative $v$, we extract from it the $v$-to-$\partial R$ distances in $G'$. 
We then use \cite{FR06}\footnote{To use \cite{FR06} we need the distances in $R$ between every pair of vertices of $\partial R$. These distaces can be computed in $O(r\log r)$ time using the MSSP algorithm~\cite{Klein05,CabelloCE13}. Over all regions, this adds only $O(\frac{n}{r}\cdot r \log r) = \Otilde(n)$ to the running time.
}  
to compute, in $\Otilde (\sqrt{r})$ time, the $v$-to-$\partial R$ distances in the entire graph $G$.  Over all $O(r^2)$ unique patterns, this takes $O(r^{2.5})$ time. 
Finally,
for each unique pattern with representative $v$ we query in $\Otilde(\sqrt{r})$ time the Voronoi diagram data structure (with weights $w(\cdot)$ being the $v$-to-$\partial R$ distances in $G$) for the furthest vertex from $v$ in $R$ in each Voronoi cell of the Voronoi diagram. 
The maximum among those gives us the furthest vertex $u_v$ from $v$ in $R$. 
The $v$-to-$u_v$ distance in $G$ can be obtained by adding the distance from $v$ to the corresponding vertex of $\partial R$ to the distance in $R$ from that vertex to $u_v$ (returned by the query to the Voronoi diagram data structure).  
Over all $O(r^2)$ unique patterns, this also takes $O(r^{2.5})$ time. 

To conclude, there are $O(n/r)$ regions and for each region $R$ we spend $\Otilde (n+r^{2.5})$ time so overall $\Otilde (n^2/r+n\cdot r^{1.5})$. Taking $r$ to be $n^{2/5}$ gives $\Otilde (n^{8/5})$ time for computing the diameter. 

\paragraph{Derandomization.}
The algorithm of \cite{MWW22} is randomized since it uses Karp-Rabin fingerprints~\cite{KarpRabin} to identify  unique patterns. Specifically, they first prove that the patterns $p^u,p^v$ of two {\em adjacent} vertices $u,v$ differ by at most two bits. Then, they traverse the graph, recording the patterns of   vertices as they are traversed. During the traversal, they maintain a complete binary tree $T$ on top of the current pattern. Each node of $T$ contains the Karp-Rabin fingerprint of its subtree (in particular, the root contains the fingerprint of the entire pattern). This way, when moving from a vertex $u$ to its neighbor $v$, at most $2 \log k$ nodes in $T$ change their fingerprint (and each of them can be computed in constant time from the two fingerprints of its children, by~\cite{KarpRabin}). 

We note here that the (randomized) Karp-Rabin fingerprints can easily be replaced with a simple  integer fingerprint scheme using a (deterministic) dictionary implemented, say, with a BST. 
The fingerprint of a leaf of $T$ is just the bit at that leaf.
The fingerprint of an internal node is determined from the pair of fingerprints of its two children $(a,b)$: 
If the pair $(a,b)$ has occurred before (i.e., the pair already appears in the BST), use the fingerprint stored for this pair in the BST; 
otherwise, store this pair in the BST and assign to it a new fingerprint (the least integer not used so far). 
Since the initial tree $T$ has $O(k)$ nodes and each new pattern changes only $2\log k$ nodes, the total number of distinct fingerprints required is $O(k^2\log k)$. 
\end{proof}

\bibliographystyle{plain}

\end{document}